\documentclass[12pt]{article}
\usepackage{amsmath, amssymb, amsthm, bm, upgreek}
\usepackage{graphicx, multirow}
\usepackage{caption, longtable}
\usepackage{enumerate}
\usepackage{natbib}
\usepackage{url, hyperref} 
\usepackage{algorithm,algpseudocode}

\newcommand{\blind}{0}

\newtheorem{theorem}{Theorem}
\newtheorem{definition}{Definition}
\newtheorem{proposition}{Proposition}

\newcommand{\indep}{\perp \!\!\! \perp}

\newcommand{\X}{{\bf X}}

\newcommand{\cP}{{\mathcal{P}}}

\newcommand{\Y}{{\bf Y}}

\newcommand{\spc}{{\mathcal S}_{Y|\X}}

\newcommand{\cms}{{\mathcal S}_{\E(Y|\X)}}
\newcommand{\R}{\mathbb{R}}

\newcommand{\A}{{\bf A}}

\newcommand{\bP}{{\bf P}}

\newcommand{\bEta}{\bm{\eta}}

\newcommand{\bLambda}{\bm{\Lambda}}

\newcommand{\E}{\mathrm{E}}

\newcommand{\bOmega}{\bm{\Omega}}

\newcommand{\bepsilon}{\bm{\epsilon}}

\newcommand{\I}{{\bf I}}

\DeclareMathOperator*{\argmin}{arg\,min}

\begin{document}

\def\spacingset#1{\renewcommand{\baselinestretch}%
{#1}\small\normalsize} \spacingset{1}


\if0\blind
{
  \title{\bf On metric choice in dimension reduction for Fr\'echet regression}
  \author{Abdul-Nasah Soale\thanks{Corresponding author: 
    abdul-nasah.soale@case.edu}, Congli Ma, Siyu Chen \hspace{.2cm}\\
    Department of Mathematics, Applied Mathematics, and Statistics,\\ Case Western Reserve University, Cleveland, OH, USA.\\
    and \\
   Obed Koomson\\
    Department of Population Health Sciences, \\
 Augusta University, Augusta, Georgia, USA}
  \maketitle
} \fi

\if1\blind
{
  \bigskip
  \bigskip
  \bigskip
  \begin{center}
    {\LARGE\bf Title}
\end{center}
  \medskip
} \fi

\bigskip
\begin{abstract}
Fr\'echet regression is becoming a mainstay in modern data analysis for analyzing non-traditional data types belonging to general metric spaces. This novel regression method is especially useful in the analysis of complex health data such as continuous monitoring and imaging data. Fr\'echet regression utilizes the pairwise distances between the random objects, which makes the choice of metric crucial in the estimation. In this paper, existing dimension reduction methods for Fr\'echet regression are reviewed, and the effect of metric choice on the estimation of the dimension reduction subspace is explored for the regression between random responses and Euclidean predictors. Extensive numerical studies illustrate how different metrics affect the central and central mean space estimators. Two real applications involving analysis of brain connectivity networks of subjects with and without Parkinson's disease and an analysis of the distributions of glycaemia based on continuous glucose monitoring data are provided, to demonstrate how metric choice can influence findings in real applications.
\end{abstract}

\noindent%
{\it Keywords:} $L_p$ spaces, brain imaging, Parkinson's disease, diabetes mellitus, continuous glucose monitoring
\vfill

\newpage
\spacingset{1.75} 

\section{Introduction\label{sec:1}}
Fr\'echet regression is becoming a mainstay in modern data analysis for analyzing such novel data types as probability distributions, positive definite matrices, networks, and other non-Euclidean data objects. Consider a response $Y$ in a general metric space $(\bOmega_Y,\rho)$ with a metric $\rho:\bOmega_Y\times\bOmega_Y \to [0, \infty)$, and a predictor $\X \in \R^k, k\geq 1$. Assume the respective marginal and conditional distributions $F_X, F_Y, F_{Y|X}$ and $F_{X|Y}$ exist and are well-defined. Then the conditional Fr\'echet mean of $Y|\X=\bm x$ is given by
\begin{align}
    m_{\oplus}(\bm x) = \argmin_{\omega\in\bOmega_Y} \E\big(\rho^2(Y,\omega)|\X=\bm x\big).
    \label{frechet_reg}
\end{align}

Because $Y$ is allowed to belong to a general metric space, including those that do not satisfy vector space properties, Fr\'echet regression is applicable in many fields. Some applications of this novel regression can be found in the study of mortality profiles in \cite{petersen2019frechet}, the influence of age on the development of brain myelination in \cite{petersen2019frechettime}, and the association between age, height, and weight on the structural connectivity of the human brain in \cite{soale2023data}, to name a few.

Following model (\ref{frechet_reg}), it is obvious that the dimensionality of $\X$ affects the estimation of the conditional Fr\'echet mean. This ``curse of dimensionality" problem can easily be tackled through sufficient dimension reduction as proposed by \citet{ying2022frechet, dong2022frechet, zhang2023dimension}, and \cite{soale2023data}. However, besides the dimensionality of $\X$, it is also obvious that $\rho$ plays a crucial role in Fr\'echet regression. And the existing literature has been clear that the choice of $\rho$ depends on the metric space under consideration. What remains unclear is how to choose $\rho$ when there are competing options. 

In most metric spaces, there are several options for $\rho$, which may vary drastically, especially in describing the geometry of the space. For instance, if say, $Y\in\R^r, r \geq 2$, the popular metric to use is the $\ell_2$-norm. However, if outliers are present in the data, a robust metric like the $\ell_1$-norm is probably better suited. Moreover, if the correlation between the responses is relevant, the Mahalanobis distance may be more appropriate. Similarly, if $Y$ belongs to the space of probability distributions, the typical choice of $\rho$ is the 2-Wasserstein distance. However, depending on the distributions under consideration, the 1-Wasserstein may be more robust. Also, other metrics such as the Hellinger distance or total variation distance will probably be more useful, especially when dealing with skewed distributions. Therefore, assuming a ``universal" metric for all scenarios for any given metric space, could lead to poor Fr\'echet regression estimates. In this paper, we will investigate the effect of metric choice on the estimation accuracy of the Fr\'echet dimension reduction in some popular metric spaces. We will present numerical comparisons of competing metrics for the same metric space under different scenarios and provide some theoretical justifications where possible. 

The rest of the paper is organized as follows. In Section 2, we provide a brief background on metric and pseudometric spaces. We also discuss some popular metric spaces and the different metric choices for each space. In Section 3, we give a general background to Fr\'echet sufficient dimension reduction, followed by an extensive simulation studies in Section 4. We present two real applications in Section 5 and conclude the paper in Section 6 with some discussion. All other technical details are relegated to the Appendix.

\section{A quick overview of metric spaces}
\subsection{General background}
\begin{definition}
Suppose $\rho:\bOmega_Y \times \bOmega_Y \to \R$ such that for all $y, y', y'' \in \bOmega_Y$,  (1) $\rho(y,y') \geq 0$;  (2) $\rho(y,y') = 0$ implies $y=y'$; (3) $\rho(y,y') = \rho(y',y)$; and (4) $\rho(y,y'') \leq \rho(y,y') + \rho(y',y'')$.

\noindent We call $\rho$ a {\bf metric} if it satisfies properties 1-4. If $\rho$ satisfies all the properties but 2, it is called a {\bf pseudometric}. 
\end{definition}

Therefore, the key difference between a metric and a pseudometric space is that unlike the metric space, there are at least two distinct elements that have a distance of zero between them in a pseudometric space. While this difference may matter under some mathematical considerations, typically, most properties of metric spaces also hold for pseudometrics. For example, regardless of whether $\rho$ is a metric or a pseudometric, we can find approximate Euclidean embeddings of $(\bOmega_Y, \rho)$ if the space is finite. 

Furthermore, it is important to note that regardless of whether we are dealing with a metric or a pseudometric space, different metric functions in a given space can have dramatic effects on the geometry of the space. For instance, in the normed space, we use the $\ell_p$ norm, denoted as $\lVert .\rVert_p$, where $p \geq 1$ is some real number. In this space, the choice of $p$ determines the geometry as shown in Figure \ref{fig:lq_norm}. We also see from Figure 1, some ``nesting" or ``ordering" behavior of the $\ell_p$ norm.
\begin{figure}[htb!]
    \centering
    \includegraphics[width=0.8\linewidth]{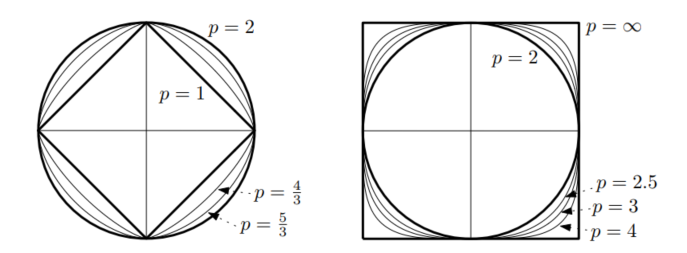}
    \caption{Shape of the unit ball in $\R^2$ for $p\in [1,\infty)$ from Section 1.3 of \cite{matousek2013}. }
    \label{fig:lq_norm}
\end{figure}

\subsection{Popular metric functions in some common metric spaces}
\subsubsection{Euclidean space}
In this space, the response set $\bOmega_Y \in\R^r$, where $r \geq 1$ is some positive integer. The most common metric in this space is the $\ell_p$ norm, which as shown in Figure \ref{fig:lq_norm} exhibits nesting.

\begin{proposition}\label{lp_ordering}
For $p \geq 1$, let $\mathbb{L}_p$ denote the class of all finite $\ell_p$ metrics on $\bOmega_Y$. Then $\mathbb{L}_{2} \subset \mathbb{L}_{p}$ for every $p$.
\end{proposition}

\noindent By Proposition \ref{lp_ordering}, $\ell_2$ is the most restrictive of the $\ell_p$ norms, and in a finite normed space, the {\it richness} of $\mathbb{L}_p$ depends on the choice of metric. In particular, for $1 \leq p_1 < p_2 \leq 2, \ \mathbb{L}_{p_2} \subset \mathbb{L}_{p_1}$ [\cite{matousek2013}].  This containment relation can have significant effects on the central space estimation as we will show later. This proposition also suggests that metrics in different metric spaces that depend on the $\ell_p$ norm will exhibit some of this behavior. 

Here, we will focus on the $\ell_1$ and $\ell_2$ norms. Thus, by Proposition \ref{lp_ordering}, in the finite Euclidean space, $\mathbb{L}_2 (Y) \subset \mathbb{L}_1(Y)$. Besides these two metrics, for $r \geq 2$, we will also consider the Mahalanobis distance, which is a pseudometric that takes into account the dependence between the responses. The Mahalanobis distance is given by
\begin{align*}
    M_d(y,y') = \sqrt{(y-y')^\top\bm{S}^{-1}(y-y')},
\end{align*}
where $\bm{S} \in \R^{r\times r}$ is the covariance matrix between $y$ and $y'$.\\ 

\subsubsection{Space of probability distributions}
In this space, $\bOmega_Y \in \mathcal{P}(Y)$, where $\mathcal{P}(Y)$ is the set of probability measures. The popular metric here is the Wasserstein distance. Let $y, y'$ be two probability measures on $\R^d, d \geq 1$ with finite $p$th-moment. The $p$-Wasserstein is given by 
\begin{align}
	W_p(y,y') = \inf_{\mathcal{F}} \big (\E\lVert y - y'\rVert_p \big )^{1/p}, \ p \geq 1,
	\label{wasserstein}
\end{align}
where $\mathcal{F}$ is the collection of all joint distributions of $y$ and $y'$. 

In particular, for $y,y' \in \R$, $W_p(.)$ can be defined in terms of quantiles as 
\begin{align}
 W_p(y,y') = \lVert F_{y}^{-1} - F_{y'}^{-1}\rVert_p = \left(\int_0^1 \big\lvert F_{y}^{-1}(\uptau)  - F_{y'}^{-1}(\uptau)\big\rvert^p d\uptau \right )^{1/p},
\label{quad_wasserstein}
\end{align}
where $F_{y}^{-1}$ and $F_{y'}^{-1}$ are the respective quantile functions of $y$ and $y'$. Thus, the Wasserstein distance is essentially the $\ell_p$ norm of the quantile functions of the distributions. This implies that we can obtain several versions of the Wasserstein metric by changing $p$. The Wasserstein distance also exhibits some ordering behavior.

\begin{proposition}\label{Wasserstein_ordering}
For $p > 1$, we have $W_1(y,y') \leq  W_p(y,y')$, $\forall y, y' \in \R$.
\end{proposition}

\noindent By Proposition \ref{Wasserstein_ordering}, the richness of the $W_p(.)$ class in finite spaces depends on the choice of $p$. This proposition also implies that $W_1(.) \leq W_2(.)$, which will be the focus of this paper. Note, as a special case, the 1-Wasserstein for $y,y'\in\R$ can also be expressed as the area between their marginal cumulative distributions, i.e., $W_1(y,y') = \int_{-\infty}^\infty \big\lvert F_{y}(\uptau)  - F_{y'}(\uptau)\big\rvert d\uptau$, see \cite{de20211}. 

The 2-Wasserstein also has a closed form for some special family of distributions. Suppose distribution $y$ has (location, scale) parameters $(\mu_y, \sigma_y)$ and  distribution $y'$ has corresponding parameters $(\mu_{y'}, \sigma_{y'})$. \cite{irpino2007optimal} showed that the squared $W_2(y, y')$ can be decomposed as
\begin{align}
    W^2_2(y, y') = \underbrace{(\mu_y - \mu_{y'})^2}_{\text{Location}} + \underbrace{(\sigma_y - \sigma_{y'})^2}_{\text{Scale}} + \underbrace{2\sigma_y\sigma_{y'} \big(1-Corr(y,y')\big)}_{\text{Shape}}.
    \label{W2_decomp}
\end{align}
Based on (\ref{W2_decomp}), $W_2(.)$ changes proportionally to the shift in the location and scale parameters, which make $W_2(.)$ sensitive to the geometry of the differences between the input distributions. Thus, $W_2(.)$ is susceptible to small outlying masses. Moreover, while the Wasserstein metric considers the support and shape of the distribution, it does appear to be better suited for comparing symmetric distributions although it can also handle asymmetric distributions. 

\begin{proposition}\label{stochastic_dominance}
Let the random variables $Y_1, Y_2 \in \R$ be such that $Y_1 \succeq_1 Y_2$, where $\succeq_1$ means stochastic dominance in the 1st order. Then, $W_1(Y_1, Y_2) = \E(Y_1) - \E(Y_2)$.
\end{proposition}

\noindent By Proposition \ref{stochastic_dominance}, when there is 1st-order stochastic dominance, $W_1(.)$ simply reduces to the difference in means. For example, for any two log-normal distributions $F_1(\mu_1, \sigma)$ and $F_2(\mu_2, \sigma)$, $F_1$ dominates $F_2$ if $\mu_1 > \mu_2$. In such scenarios, using the $W_1(.)$ is preferable to $W_2(.)$ as the latter will likely introduce some noise in finite samples based on the decomposition in (\ref{W2_decomp}).

Another metric we will consider in the probability space is the Hellinger distance. The Hellinger distance is given by
\begin{align}
    H(y,y') = 1 - \int \sqrt{P_y(\omega)P_{y'}(\omega)}  \partial \omega,
\end{align}
where $P_y$ and $P_{y'}$ are continuous probability distributions. For discrete distributions, we replace the integral with summation. Unlike the Wasserstein distance, the Hellinger distance is bounded between 0 and 1, and it puts more emphasis on tail variations than the center. \\

\subsubsection{Space of symmetric positive definite matrices}
In this space, $\bOmega_Y \in {\bm S}_r^{+} $, where ${\bm S}_r^{+}$ denote the space of $r\times r$ symmetric positive definite matrices. Here, the popular metric is the Frobenius norm. For an $m\times n$ matrix $\A=[a_{ij}]$, the Frobenius norm is given by $\lVert \A \rVert_F = \left (\displaystyle\sum_{i=1}^m\sum_{j=1}^n a^2_{ij}\right )^{1/2}$. Because each matrix entry is squared, the Frobenius norm is susceptible to the influence of outliers and noise. Therefore, we also consider two derivatives of the norm based on the scaled versions of the original matrices. Our choice of scaling functions are the ``square root" i.e., Cholesky decomposition and the logarithm.

Therefore, for response matrices $y, y' \in \bm S_r^{+}$, we consider the following norms:
\begin{align*}
    \text{Frobenius norm (F):} & \ \lVert  y - y' \rVert_F, \\
    \text{Cholesky Frobenius norm (Chol F):} &  \ \lVert  \mathrm{Chol}(y) - \mathrm{Chol}(y') \rVert_F, \\
    \text{Log Frobenius norm (Log F):} & \ \lVert  \log_m(y) - \log_m(y') \rVert_F,
\end{align*}
where $\mathrm{Chol}(\A)=\L\L^\top$ is the Cholesky decomposition of matrix $\A$ for some lower triangular matrix $\L$ and $\log_m(.)$ is the matrix logarithm. 

\begin{proposition}\label{spdm_prop}
For any $y, y' \in  \bm S_r^{+}$,  $\lVert  y - y' \rVert_F   \leq  \lVert  \mathrm{Chol}(y) - \mathrm{Chol}(y') \rVert_F$  if the $\max\{\lVert \L_y\rVert_F, \lVert \L_{y'}\rVert_F \} \geq 0.5$, where $y=\L_y\L_y^\top$ is the Cholesky decomposition of $y$.
\end{proposition}

\subsubsection{Space of networks}
In this space, $\bOmega_Y \in \mathcal{G}$, the collection of networks or graphs. A graph $G=(V,E)$, where $V=\{\text{vertices or nodes}\}$ and $E=\{\text{edges or links}\}$. For this space, we consider the centrality distances of \cite{roy2014modeling}. The centrality distance between two networks $G_1=(V_1,E_1)$ and $G_2=(V_2,E_2)$ is given by 
    \begin{align}
        C_\rho(G_1,G_2) = \displaystyle\sum_{v\in V} |C(G_1,v) - C(G_2,v)|,
        \label{cent_dist}
    \end{align}
where $C(G_i,v)$ is some centrality measure of network $G_i$ for node $v$. Thus, $C_\rho(G_1,G_2)$ is essentially the $\ell_1$ norm of the difference in centralities. There are various centrality measures but will only focus on the degree centrality, denoted as $C_d$ and the closeness centrality denoted as $C_c$. We expect $C_d$ to be influenced by the degree distribution of the networks and $C_c$ to depend on the existence of hubs and clusters.

Another metric we will consider in this space is the diffusion distances of \cite{hammond2013graph} given by
\begin{align}
    D_d(G_1,G_2) = \max_\uptau \big(\lVert \mathrm{expm}(-\uptau\L_1) - \mathrm{expm}(-\uptau\L_2)\rVert_F^2\big)^{1/2},
\label{diff_dist}
\end{align}
where $\L_i$ is the Laplacian of graph $G_i$, $\uptau \in (0,1)$, and $\mathrm{expm}(\A)$ is the exponential of matrix $\A$ given by $\mathrm{expm}(\A) = \displaystyle\sum_{n=0}^\infty \cfrac{\A^n}{n!}$. Unlike the degree centrality distance, the diffusion distance can account for edge weight and other network characteristics beside node characteristics.

We will now proceed to show how the competing metrics in these popular metric spaces influence the dimension reduction subspace for Fr\'echet regression in the next section.

\section{Fr\'echet sufficient dimension reduction}
\subsection{A brief background}
The estimates of the global and local conditional Fr\'echet means in model (\ref{frechet_reg}) are based on the least squares assumption, which can be restrictive in many scenarios, the obvious one being the recovery of a single direction. Fortunately, we can minimize some of the restrictions by estimating the central mean space instead. 

A subspace $\bEta \in \R^{k\times d}$ is a mean dimension reduction subspace for the regression between $Y$ and $\X$ if it satisfies 
\begin{align}
    \E(Y|\X) = \E(Y|\bEta^\top\X),
    \label{cms}
\end{align}
where $\bEta^\top\bEta = \I_d$, with $d < k$ as the structural dimension. Although $\bEta$ is not identifiable, the intersection of all such $\bEta$ that satisfies (\ref{cms}) is shown to exist and is called the central mean space, denoted as $\cms = \mathrm{span}(\bEta)$ [\cite{cook1996graphics} and \cite{yin2008successive}].

While the conditional mean is the focus of most applications, the relationship between $Y$ and $\X$ sometimes goes beyond the mean. To account for regression relationships beyond the mean, we can estimate a broader subspace called the {\it central space}. In general,  $\bEta\in\R^{k\times d}$ is a dimension reduction subspace if it is such that
\begin{align}
	Y \indep \X |\bEta^\top\X,
	\label{sdr}
\end{align}
where $\indep $ means statistical independence and $\bEta^\top\bEta = \I_d$. Again, $\bEta$ is not unique in (\ref{sdr}). However, the central space denoted as $\spc = \mathrm{span}(\bEta)$ is identifiable if it exists, under the same conditions outlined in \cite{cook1996graphics} and \cite{yin2008successive}. The central space contains the central mean space, i.e., $\cms \subseteq \spc$. 

In dimension reduction for Fr\'echet regression, we may not be able to use $Y$ directly in the estimation as $Y$ may not satisfy the vector space properties, such as inner products on which most SDR estimators are based. Therefore, we utilize the embeddings of $(\bOmega_Y, \rho)$ as the surrogate for $Y$.

\begin{definition}[Metric embedding]
For some metric spaces $(\bOmega_Y,\rho_Y)$ and $(\bOmega_z,\rho_Z)$, consider the mapping $g:(\bOmega_Y,\rho_Y) \to (\bOmega_Z,\rho_Z)$ such that $\forall y,y'\in\bOmega_Y$ we have
\begin{itemize}
    \item[(a)] $\rho_Y(y,y') = \rho_Z\big(g(y),g(y')\big)$;
    \item[(b)] $(1-\epsilon)\rho_Y(y,y') \leq \rho_Z\big(g(y),g(y')\big) \leq (1+\epsilon)\rho_Y(y,y')$, for some $0 < \epsilon < 1$; 
    \item[(c)] $\rho_Y(y,y') \leq C\rho_Z\big(g(y),g(y')\big)$, for some constant $C$.
\end{itemize}
The mapping $g$ is called an {\bf isometric embedding} if it satisfies (a). If $g$ only holds for (b), it is called an {\bf $\epsilon$-almost isometry}. Lastly, the mapping in  (c) is called a {\bf $C$-Lipschitz} map.
\end{definition}

\noindent Depending on the choice of embedding of $(\bOmega_Y, \rho)$ employed, we can broadly classify existing Fr\'echet SDR techniques into kernel and kernel-free methods. 

{\bf Kernel methods:} The methods in this class assume that there is a continuous isometric embedding $g :\bOmega_Y \mapsto \mathcal{H}_Y$, where $\mathcal{H}_Y$ is a Hilbert space dense in the family of continuous real-valued functions of $\bOmega_Y$. Thus, for any $y, y' \in \bOmega_Y$ the surrogate response $\widetilde Y = \kappa\left(\rho(y, y')\right)$ can be obtained, where $\kappa(.)$ is a positive definite \emph{universal} kernel and $\rho$ is an appropriate metric on $\bOmega_Y$.

To satisfy the continuous embedding condition, it is sufficient for $(\bOmega_Y, \rho)$ to be complete and separable, which holds if $\rho$ is of the {\it negative type}. However, whether or not $\kappa(.)$ is universal depends on the metric space being considered. For instance, in the space of univariate probability distributions, the popular Gaussian radial basis function kernel is universal. However, the same is not true for spheres. Therefore, for kernel Fr\'echet SDR estimators, the choice of metric and kernel type is specific to the metric space. Examples of methods that fall into this category include the Fr\'echet kernel sliced inverse regression of \cite{dong2022frechet} and the extension of existing SDR for Euclidean response to Fr\'echet SDR for random responses by \cite{zhang2023dimension}. 

{\bf Kernel-free methods:} This class of estimators utilize approximate Euclidean embeddings of $\bOmega_Y$, i.e., $g:\bOmega_Y \to \R$. Thus, for $y,y'\in \bOmega_Y$, the surrogate response is given by $\widetilde Y = T(\rho(y, y'))$, where $T$ is a linear transformation. This linear transformation is simply a random univariate projection of the pairwise distance matrix in finite samples. Because the embedding does not have to be strictly isometric, several metric functions including true metrics and pseudometrics are applicable. An example of this technique can be found in \cite{soale2023data}.

By the transformational properties of the $\spc$, the containment property of the metrics discussed in Section 2, can have some influence on the $\cms$ and $\spc$ estimates, regardless of the estimation technique employed. Moreover, for the $\cms$, the estimates depends on the concentration of the embedding. That is, if the embedding does not concentrate around the center of the ``measure", the $\cms$ methods will not be able to recover the regression subspace. 

\begin{theorem}\label{containment}
Let $\rho_1$ and $\rho_2$ be metrics defined on $\bOmega_Y$. 
Define functional spaces $\mathcal{L}_{\rho_1}(\bOmega_Y) = \big\{f: f(\rho_1(\bOmega_Y)) < \infty \big\}$ and  $\mathcal{L}_{\rho_2}(\bOmega_Y) = \big\{f: f(\rho_2(\bOmega_Y)) < \infty \big\}$ for some measurable function $f$ with finite moments. If $\mathcal{L}_{\rho_1}(\bOmega_Y) \subseteq \mathcal{L}_{\rho_2}(\bOmega_Y)$, then $\mathcal{S}_{\rho_2} \subseteq \mathcal{S}_{\rho_1}$, where $\mathcal{S}_\rho$ is the dimension reduction subspace for the regression between $f(\rho(\bOmega_Y))$ and $\X$.
\end{theorem}

\noindent Theorem \ref{containment} emphasizes the importance of choosing the appropriate metric as it relates to information about the central space such as dimensionality, unbiasedness, Fisher consistency, and exhaustiveness. See \cite{li2018sufficient} for more details on these properties.

\begin{theorem}\label{lipschitz_cms}
Let $\rho_1$ and $\rho_2$ be metrics defined on $\bOmega_Y$. If $\rho_1$ concentrates around the mean while $\rho_2$ does not, then $\rho_1$ will yield a better estimate of $\cms$ than $\rho_2$.
\end{theorem}

\noindent For Theorem \ref{lipschitz_cms} to be satisfied, it suffices for the response embedding to be Lipschitz, which may not hold for all metrics defined on the same space. For instance, the Wasserstein distance between two normal distributions is Lipschitz but the Hellinger distance is not. Some metrics may also satisfy the Lipschitz condition locally. Moreover, for metrics without upper bounds, the embedding on scaled inputs is more likely to concentrate around the mean than the embedding based on the actual values. For instance, we expect the Frobenius norm of the difference between logged positive definite matrices to concentrate around the mean faster than the Frobenius norm between the actual matrices. Thus, it is important to evaluate metrics on case-by-case basis. The central mean space estimators are especially powerful when the concentration of the metric is Gaussian or approximately Gaussian.

\section{Numerical Studies}
In this Section, we demonstrate the effect of metric choice on Fr\'echet SDR estimation using synthetic data. The estimators considered are the Fr\'echet ordinary least squares (FOLS) and the Fr\'echet sliced inverse regression (FSIR) of \cite{zhang2023dimension}. We also include the surrogate-assisted ordinary least squares (sa-OLS) and the surrogate-assisted sliced inverse regression (sa-SIR) of \cite{soale2023data}. FOLS and sa-OLS are specifically designed to estimate the $\cms$ while FSIR and sa-SIR are for estimating the $\spc$. The performance of each method is measured in terms of the estimation error $\Delta$ defined as
\begin{align}
	\Delta = \lVert \bP_{\bEta} - \bP_{\hat\bEta} \rVert_F,
\end{align}  
where $\bP_{\bEta} = \bEta(\bEta^\top \bEta)^{-1}\bEta^\top$. Thus, smaller values of $\Delta$ indicate higher accuracy. 

We generate the random samples $\{(\X_i, Y_i)\}_{i=1}^n$ using the same data generation process for the predictor, but vary the response models based on the metric space. We fix $(n,p) = (100, 10)$ and $(500, 20)$ and generate the predictor $\X \overset{i.i.d.}{\sim} t_{20}(\bm 0, \I_p)$. 
Next, we let $\beta_1^\top = (1,1,0,\ldots, 0)/\sqrt{2}\in \R^p$ and $\beta_2^\top = (0,0,1,1,0\ldots, 0)/\sqrt{2}\in \R^p$, and generate the responses as follows. For easy replication, the number of slices is fixed at 5 for both FSIR and sa-SIR. We also fixed the number of projection vectors $N=1000$ for sa-OLS and sa-SIR.

\subsection{Euclidean responses}
The response is generated in $\R^2$ as follows:
\begin{enumerate}
\item [I.] $\Y = (Y_1, Y_2)^\top$ with 
$Y_1 = \sin(\beta_1^\top\X) + 0.5\epsilon_1$ and $Y_2 = \cos(\beta_1^\top\X) + 0.5\epsilon_2$,
\item [II.] $\Y = (Y_1, Y_2)^\top$ with 
$Y_1 = X_1 + X_2^3 + \epsilon_1$ and $Y_2 = 0.5e^{X_1 +0.2} + \epsilon_2$, 
\end{enumerate}
where $\bepsilon = (\epsilon_1, \epsilon_2)^\top \overset{i.i.d.}{\sim} N_2(\bm 0, \I_2)$.

In model I, the true basis $\eta = \beta_1$. From Table \ref{tab:Euc_sims}, we see that the $\cms$ estimates based on the Mahalanobis distance outperforms the estimates based on the $\ell_1$ and $\ell_2$ norms. The same pattern can be seen in the $\spc$ estimates but the difference is less pronounced. This is not surprising because although the responses are uncorrelated at the population level, it is possible that the samples are somewhat correlated. In fact, $\sin(\bm x)$ and $\cos(\bm x)$ are uncorrelated if say, $\bm x \overset{i.i.d}{\sim} Uniform(0, 2\pi)$. However, for $\bm x \overset{i.i.d}{\sim} Uniform (0, \pi/2)$, the correlation between them is -0.5.  

The true basis in model II is $\bEta = (e_1, e_2)$, where $e_1^\top=(1,0,\ldots,0) \in\R^p$ and $e_2^\top=(0,1,\ldots,0) \in\R^p$. Here, $Y_1$ and $Y_2$ are correlated for positive values of $X_1$. Thus, we expect the Mahalanobis distance to outperform the other metrics. $Y_2$ is prone to  outliers, which is a plausible reason why the estimates based on the $\ell_1$ norm appear to be slightly better than those based on the $\ell_2$ norm. FOLS and FSIR were not implemented for the Euclidean response in the original paper of \cite{zhang2023dimension}, and are thus omitted.

\begin{center}
\small
\setlength{\tabcolsep}{4pt}
\LTcapwidth=\textwidth
\begin{longtable}{ccccccc}
\caption{Mean (SD) of $\Delta$ based on 500 random samples with Euclidean response}
\label{tab:Euc_sims} \\
\hline
Model & $(n,p)$ & Metric  &  FOLS & sa-OLS & FSIR & sa-SIR \\ 
\hline
\multirow{6}{*}{I} & \multirow{3}{*}{(100, 10)} & $\ell_1$ & \multirow{3}{*}{NA} & 0.5481 (0.0077) & \multirow{3}{*}{NA} & 0.4961 (0.0096) \\
 & & $\ell_2$ &  & 0.4904 (0.0072) &   & 0.4675 (0.0089) \\
 & & $M_d$ &  & 0.4071 (0.0061) &  & 0.4429 (0.0079) \\ \cline{2-7} 
 & \multirow{3}{*}{(500, 20)} & $\ell_1$ & \multirow{3}{*}{NA} & 0.3621 (0.0031) & \multirow{3}{*}{NA} & 0.2721 (0.0026) \\
& & $\ell_2$ & & 0.3200 (0.0027) & & 0.2608 (0.0025) \\
&  & $M_d$ &  & 0.2712 (0.0025) & & 0.2526 (0.0024) \\ \cline{1-7}

\multirow{3}{*}{II} & \multirow{3}{*}{(100, 10)} & $\ell_1$ & \multirow{3}{*}{NA} & 1.1544 (0.0104) & \multirow{3}{*}{NA} & 1.0696 (0.0119) \\
& & $\ell_2$ & & 1.2059 (0.0095)  &  & 1.0859 (0.0117) \\
& & $M_d$ &  & 0.9416 (0.0099) &  & 0.9874 (0.0119) \\ \cline{2-7} 
\quad \\

\multirow{3}{*}{II} & \multirow{3}{*}{(500, 20)} & $\ell_1$ & \multirow{3}{*}{NA} & 1.1657 (0.0087) & \multirow{3}{*}{NA} &  0.5635 (0.0047)\\
& & $\ell_2$ &  &  1.2343 (0.0074) &  &  0.5800 (0.0049)\\
& & $M_d$ & & 0.8112 (0.0090) &  & 0.5220 (0.0038) \\ \hline
\end{longtable}
\end{center}

\subsection{Space of univariate distributions}
Here, the responses are generated as samples from a distribution. For $i=1,\ldots,n$, we generate the samples as follows:
\begin{enumerate}
     \item[III.] $Y_i \in \R^{50} \overset{i.i.d}{\sim} Poisson\big( \exp(\beta_1^\top\X_i)\big)$,
     \item[IV.] $Y_i \in \R^{50} \overset{i.i.d}{\sim} \log-normal\big(\beta_1^\top\X_i, 1\big)$,
     \item[V.] $Y_i \in \R^{50} \overset{i.i.d.}{\sim} Gamma\big(\alpha(\X_i), 3 \big)$,  where $\alpha(\X_i) =  e^{\beta_1^\top\X_i} $.
\end{enumerate}

Here, all three models III, IV, and V are single-index with $\eta = \beta_1$. From Table \ref{tab:Dist_sims}, the Hellinger distance shows significant improvement in the $\cms$ estimates for both models III and V over the Wasserstein metrics. For model III, the Wasserstein distances are not globally Lipschitz but may be locally Lipschitz when $\beta_1^\top\X < 0$. Thus, we do not expect $W_1$ and $W_2$ to concentrate around the mean by Theorem \ref{lipschitz_cms}. On the other hand, the Hellinger distance concentrates around the center. Moreover, the Hellinger distance is better suited for capturing the tail variations compared to the Wasserstein metrics. This is one of the reasons for the difference we see for the $\cms$ estimates for model III. However, we see less difference in the $\spc$ estimates as the central space is not restricted to the mean.

In model IV, the Hellinger distance performs poorly compared to Wasserstein metrics in both the $\cms$ and $\spc$ estimates. The 1-Wassertein performs slightly better than the 2-Wasserstein because any pair of log-normal distributions with same scale parameters exhibit stochastic dominance. Thus, based on Proposition \ref{stochastic_dominance}, $W_1(.)$ captures just the difference in the means while $W_2(.)$ may introduce some noise from the difference in sample scale parameters and shape. The log-normal distributions are more symmetric, which makes the Wasserstein a better metric for capturing location differences than the Hellinger distance.

Lastly, for the same scale, the Gamma distribution is more right-skewed for smaller values of the shape parameter. Thus, the variations between the distributions are more likely to be in the right tails, which is better captured by the Hellinger distance. This affects the concentrations of the metrics at the mean as seen from the $\cms$ estimates.

\begin{center}
\small
\setlength{\tabcolsep}{4pt}
\LTcapwidth=\textwidth
\begin{longtable}{ccccccc}
\caption{Mean (SD) of $\Delta$ based on 500 random samples for distributions as response}
\label{tab:Dist_sims}\\
\hline
Model & $(n,p)$ & Metric  &  FOLS & sa-OLS & FSIR & sa-SIR \\ 
\hline
\multirow{6}{*}{III}  & \multirow{3}{*}{(100, 10)} & $W_1$ & 0.6583 (0.0099) & 0.6560 (0.0099) & 0.2614 (0.0053) & 0.2280 (0.0031) \\
& & $W_2$ & 0.6425 (0.0100) & 0.6397 (0.0100) & 0.2629 (0.0053) & 0.2271 (0.0032) \\
& & $H$  & 0.3272 (0.0045)  & 0.3129 (0.0044) & 0.2843 (0.0044) & 0.2817 (0.0044)  \\ \cline{2-7} 
& \multirow{3}{*}{(500, 20)} & $W_1$ & 0.6807 (0.0095) & 0.6868 (0.0095) & 0.1677 (0.0057) & 0.1359 (0.0011) \\
& & $W_2$ & 0.6669 (0.0097) & 0.6728 (0.0096) & 0.1671 (0.0056) & 0.1358 (0.0011) \\
& & $H$ & 0.2044 (0.0018) & 0.1939 (0.0015) & 0.1545 (0.0013) & 0.1549 (0.0014) \\ \cline{1-7}

\multirow{6}{*}{IV}  & \multirow{3}{*}{(100, 10)} & $W_1$ & 0.6697 (0.0099) & 0.6653 (0.0099) & 0.2769 (0.0069) & 0.2293 (0.0032) \\
& & $W_2$ & 0.7063 (0.0098) & 0.7016 (0.0099) & 0.3245 (0.0089) & 0.2444 (0.0034) \\
& & $H$ & 1.3385 (0.0046) & 1.3405 (0.0044) & 1.3411 (0.0044) & 1.3383 (0.0045) \\ \cline{2-7} 
& \multirow{3}{*}{(500, 20)} & $W_1$ & 0.6992 (0.0095) & 0.7015 (0.0095) & 0.1908 (0.0074)      & 0.1384 (0.0011) \\
& & $W_2$ & 0.7356 (0.0094) & 0.7365 (0.0094) & 0.2435 (0.0105) & 0.1458 (0.0012) \\
& & $H$ & 1.3752 (0.0023) & 1.3788 (0.0021) & 1.3807 (0.0021) & 1.3811 (0.0021)  \\ \cline{1-7} 

\multirow{6}{*}{V} & \multirow{3}{*}{(100, 10)} & $W_1$ & 0.6250 (0.0103) & 0.6547 (0.0100) & 0.2502 (0.0042) & 0.2257 (0.0030) \\
& & $W_2$ & 0.6124 (0.0103) & 0.6394 (0.0101) & 0.2611 (0.0042) & 0.2266 (0.003) \\
& & $H$ & 0.377 (0.0053) & 0.3512 (0.0045) & 0.4183 (0.0078) & 0.3762 (0.0071) \\ \cline{2-7} 
& \multirow{3}{*}{(500, 20)} & $W_1$ & 0.6338 (0.0101) & 0.6859 (0.0095) & 0.1543 (0.0037) & 0.1351 (0.0011) \\
& & $W_2$ & 0.6216 (0.0102) & 0.6721 (0.0097) & 0.1608 (0.0037) & 0.1356 (0.0011) \\
& & $H$ & 0.2143 (0.0018) & 0.2113 (0.0018) & 0.1995 (0.0017) & 0.1831 (0.0015) \\  \hline
\end{longtable}
\end{center}

\subsection{Space of positive definite matrices}
For $i=1,\ldots,n$, generate the positive definite matrix responses as follows:
\begin{enumerate}
	\item[VI.] $Y_i = \mathrm{expm}\begin{pmatrix}
	    1 & \rho(\X_i) \\ \rho(\X_i) & 1
	\end{pmatrix} + \epsilon\I_2$ with $\rho(\X_i) = \cfrac{e^{\beta_1^\top\X_i} - 1 }{e^{\beta_1^\top\X_i} + 1}$,
       \item[VII.] $Y_i = 2\bLambda\bLambda^\top + \epsilon\I_2$ with $\bLambda = (e^{\beta_1^\top\X_i}, e^{\beta_2^\top\X_i})$,
       \item[VIII.] $Y_i = \begin{pmatrix}
	    1 & \rho_1(\X_i) & \rho_2(\X_i) \\ \rho_1(\X_i) & 1 & \rho_1(\X_i) \\
     \rho_2(\X_i) & \rho_1(\X_i) & 1 
	\end{pmatrix} + \epsilon\I_3$, \\
 where $\rho_1(\X_i) = 0.5\left(\cfrac{e^{\beta_1^\top\X_i}}{e^{\beta_1^\top\X_i} + 1}\right)$, $\rho_2(\X_i) = 0.8\sin(\beta_2^\top\X_i)$, and $\epsilon = 0.10$.
\end{enumerate}

In model VI, $\eta = \beta_1$ while in models VII and VIII, $\bEta = (\beta_1, \beta_2)^\top$. We observe in model VI that the Frobenius norm based on the actual matrices performs slightly better than the ones based on the scaled matrices in the $\cms$ estimate. However, for the entire central space, the Cholesky and log Frobenius norms yield better estimates. Unsurprisingly, Frobenius norm based on the scaled matrices, i.e., log and Cholesky dominates in model VII for estimating both the $\cms$ and $\spc$--with the difference more striking for the $\cms$ estimation. The domination of log Frobenius is expected as the Frobenius norm of the exponentiated matrix values are less likely to concentrate around the mean. A similar pattern is observed in model VIII. The less dispersion in the scaled Frobenius norms is a plausible reason why they yield better estimates for the $\spc$ estimates.

\begin{center}
\small
\setlength{\tabcolsep}{4pt}
\LTcapwidth=1.2\textwidth
\begin{longtable}{ccccccc}
\caption{Mean (SD) of $\Delta$ based on 500 random samples for SPD matrices as response}
\label{tab:SPDM_sims}\\
\hline
Model & $(n,p)$ & Metric  &  FOLS & sa-OLS & FSIR & sa-SIR \\ 
\hline
\multirow{6}{*}{VI} & \multirow{3}{*}{(100, 10)} & F & 0.1613 (0.0240) & 0.1405 (0.0152) & 0.2346 (0.0318) & 0.2152 (0.0255) \\
& & Chol F & 0.1809 (0.024) & 0.1434 (0.0152) & 0.22 (0.0286) & 0.2081 (0.0239)  \\
& & Log F & 0.1842 (0.0246) & 0.1407 (0.0145) & 0.2193 (0.0285) & 0.2064 (0.0236) \\ \cline{2-7} 
& \multirow{3}{*}{(500, 20)} & F & 0.0941 (0.0054) & 0.0837 (0.0042) & 0.1195 (0.0069) & 0.112 (0.0065) \\
& & Chol F & 0.1118 (0.0066) & 0.0878 (0.0049) & 0.1128 (0.0062) & 0.1076 (0.0061) \\
& & Log F & 0.1142 (0.0067)  & 0.0857 (0.0048) & 0.1123 (0.0061)  & 0.1069 (0.0061) \\ \cline{1-7} 

\multirow{6}{*}{VII}  & \multirow{3}{*}{(100, 10)} & F & 0.9441 (0.1459) & 0.9488 (0.1432) & 0.7083 (0.1791) & 0.4721 (0.0413) \\
& & Chol F & 0.6362 (0.0973) & 0.6463 (0.1017) & 0.4714 (0.0463) & 0.4365 (0.0392) \\
& & Log F & 0.4753 (0.0339) & 0.4748 (0.0341)  & 0.5598 (0.048)  & 0.4295 (0.0414) \\ \cline{2-7} 
 & \multirow{3}{*}{(500, 20)} & F & 0.8198 (0.068) & 0.8332 (0.0742) & 0.2653 (0.0157) & 0.2418 (0.0043) \\
& & Chol F & 0.4005 (0.0239) & 0.4231 (0.0265) & 0.2345 (0.0068) & 0.2185 (0.0045) \\
& & Log F & 0.2848 (0.0214) & 0.2932 (0.0224) & 0.2798 (0.0146)  & 0.2119 (0.0084) \\ \cline{2-7} 

\multirow{6}{*}{VIII} & \multirow{3}{*}{(100, 10)} & F & 0.7428 (0.1689) & 0.7203 (0.1661) & 0.9672 (0.1645) & 0.5411 (0.0974) \\
& & Chol F & 0.6518 (0.2004) & 0.6837 (0.1932) & 0.8349 (0.1527) & 0.455 (0.0281) \\
& & Log F & 0.6345 (0.1779)  & 0.6226 (0.2052) & 0.9908 (0.1348) & 0.4389 (0.0206) \\ \cline{2-7} 
& \multirow{3}{*}{(500, 20)} & F & 0.5027 (0.0772) & 0.5259 (0.0752) & 0.5051 (0.0839) & 0.2509 (0.0257) \\
& & Chol F & 0.3046 (0.0623) & 0.3294 (0.0702) & 0.4037 (0.0705) & 0.2195 (0.0201) \\
& & log F & 0.2388 (0.046) & 0.2439 (0.0484) & 0.3391 (0.0595) & 0.2047 (0.0175) \\ \hline
\end{longtable}
\end{center}

\subsection{Space of networks}
Networks of size 50 are generated based on three network models: Erd\"os-R\'enyi, the preferential attachment, and the stochastic block model (SBM).  For $i=1,\ldots,n$, we generate the networks as follows:

\begin{enumerate}
  \item[IX.]  $Y_i\sim$ an Erd\"os-R\'enyi model with edge probability $p=\cfrac{e^{\sin(\beta_1^\top\X_i)}}{1+e^{\sin(\beta_1^\top\X_i)}}$. 

  \item[X.] $Y_i\sim$ a preferential attachment model with power $\gamma=1 + 0.8\left(\cfrac{e^{\beta_1^\top\X_i}}{1+e^{\beta_1^\top\X_i}}\right)$. 

   \item[XI.] $Y_i\sim$ a weighted stochastic block model with three clusters prior probabilities 0.30, 0.45, 0.25. The edge weights  $\overset{i.i.d.}{\sim} Poisson \big(\lambda_{j\ell} + \lceil e^{\beta_1^\top\X_i} \rceil \big)$,  where \\ $\lambda_{j\ell} = \begin{cases} 6 \text{ if } j=\ell \\ 1 \text{ if } j\neq \ell,  \end{cases}$ for $j,\ell = 1,2,3$.\\
    See \cite{leger2016blockmodels} for more on this model.
\end{enumerate}

Models IX, X, and XI are all single index models with $\eta=\beta_1$. Besides the networks based on the Erd\"os-R\'enyi model, the degree centrality distances perform worse in the preferential attachment networks and the networks based on stochastic block models. This is not surprising as the degree distributions for the preferential attachment networks follow the power law, which is highly skewed. The closeness centrality distance shows improvement over the other network distances in the stochastic block models, where communities are highly connected. The diffusion distance appears to be better suited for preferential attachment networks. As the preferential attachment and stochastic block models are more likely to mimic real networks, we suggest using the closeness centrality and diffusion distance in real applications. Network responses were not considered in the original paper on \cite{zhang2023dimension}, and are thus omitted.

\begin{center}
\small
\setlength{\tabcolsep}{4pt}
\LTcapwidth=1.2\textwidth
\begin{longtable}{ccccccc}
\caption{Mean (SD) of $\Delta$ based on 500 random samples for networks as response}
\label{tab:net_sims}\\
\hline
Model & $(n,p)$ & Metric  &  FOLS & sa-OLS & FSIR & sa-SIR \\ 
\hline
\multirow{3}{*}{IX} & \multirow{3}{*}{(100, 10)} & $C_d$ & \multirow{3}{*}{NA}  & 0.2590 (0.0049) & \multirow{3}{*}{NA}  & 0.2717 (0.0052) \\
& & $C_c$ & & 0.2603 (0.0049) & & 0.2731 (0.0052) \\
& &  $D_d$ & & 0.3043 (0.0051) & & 0.2707 (0.0052) \\ \cline{2-7} 
\quad \\
\multirow{3}{*}{IX} & \multirow{3}{*}{(500, 20)} & $C_d$ & \multirow{3}{*}{NA}  & 0.1709 (0.0017) & \multirow{3}{*}{NA}  & 0.1725 (0.0018) \\
&  & $C_c$ & & 0.1725 (0.0017) & & 0.1729 (0.0018) \\
&  & $D_d$ & & 0.2002 (0.0019) & & 0.1720 (0.0018)  \\ \cline{1-7}

\multirow{6}{*}{X} & \multirow{3}{*}{(100, 10)} & $C_d$ & \multirow{3}{*}{NA}  & 0.8819 (0.0117)  & \multirow{3}{*}{NA}  & 1.0673 (0.0127) \\
&  & $C_c$ & & 0.6321 (0.0067) &  & 0.6976 (0.0094) \\
&  & $D_d$ & & 0.4907 (0.0053) &  & 0.9635 (0.0152) \\ \cline{2-7} 
& \multirow{3}{*}{(500, 20)} & $C_d$ & \multirow{3}{*}{NA}  & 0.4922 (0.004) & \multirow{3}{*}{NA}  & 0.5105 (0.0058) \\
& & $C_c$ & & 0.4327 (0.0033) & & 0.425 (0.0033)  \\
& & $D_d$ & & 0.3236 (0.0024) & & 0.3301 (0.0025)  \\ \cline{1-7} 
 
\multirow{6}{*}{XI} & \multirow{3}{*}{(100, 10)} & $C_d$ & \multirow{3}{*}{NA}  & 0.4965 (0.0076) & \multirow{3}{*}{NA}  & 0.3033 (0.0038)  \\
& & $C_c$ &  & 0.2704 (0.0032) & & 0.2782 (0.0034)  \\
& & $D_d$ & & 0.3981 (0.0046) &  & 0.2849 (0.0035)  \\ \cline{2-7} 
& \multirow{3}{*}{(500, 20)} & $C_d$ & \multirow{3}{*}{NA}  & 0.4096 (0.0063) & \multirow{3}{*}{NA}  & 0.1875 (0.0014) \\
& & $C_c$ & & 0.1737 (0.0013) & & 0.1742 (0.0013)  \\
& & $D_d$ & & 0.2637 (0.0019) & & 0.1807 (0.0013)  \\ \hline
\end{longtable}
\end{center}

\section{Real Applications}
In this section, we present two applications with responses belonging to two different metric spaces. In the first application we focus on inference and in the second we focus on prediction. 

\subsection{Functional brain connectivity network data}
Recent advances in functional magnetic resonance imaging (fMRI), has allowed us to capture in-depth structural and functional activities of the human brain. By leveraging such imaging data, we can better understand certain neurological diseases such as Alzheimer's and Parkinson's disease (PD). In this study, we analyze the correlations between different regions of the brain generated using the HarvardOxford Parcellation method. The analysis data consists of brain connectivity networks in the Neurocon dataset released by \cite{badea2017exploring}. There are 41 subjects in the dataset, 26 of whom have Parkinson's disease. Details about the preprocessing of the images as well as the codes are publicly available at
\href{https://github.com/brainnetuoa/data_driven_network_neuroscience}{https://github.com/brainnetuoa/data\_driven\_network\_neuroscience} and  \href{https://doi.org/10.17608/k6.auckland.21397377}{https://doi.org/10.17608/k6.auckland.21397377}. 

Our response of interest is the correlation between the 48 regions of interest (ROI) of the subject's brain. Figure \ref{fig:brain_corr} shows the comparison between the correlation plot of a subject with and without Parkinson's disease.

\begin{figure}[htb!]
    \centering
    \includegraphics[width=\linewidth]{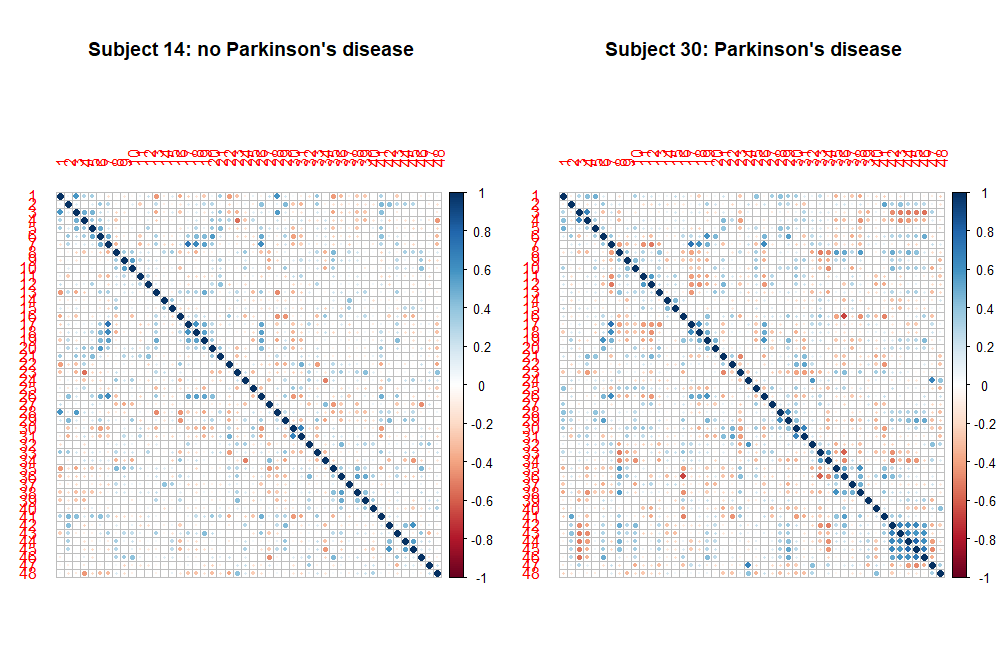}
    \caption{Correlation plot between the 48 different regions of the brain for a subject with and without Parkinson's disease}
    \label{fig:brain_corr}
\end{figure}

We take as predictors, the recorded age, sex (F/M), and diagnosis (PD/control). The age range is between 45 and 86, thus, to avoid undue influence from outliers, we use the natural logarithm of the age variable in the analysis. Next, we proceed to find the central mean space $\cms$ estimates using the the Fr\'echet OLS (FOLS) and surrogate-assisted OLS (sa-OLS) under different metrics, i.e., the Frobenius (F), Cholesky Frobenius (Chol F), and log Frobenius (log F) metrics. In particular, we want to compare the average difference in correlation structure between the subjects diagnosed with Parkinson's disease vs those without the disease after controlling for age and sex. The results are provided in Table \ref{fMRI_brain}.

\begin{table}[htb!]
\caption{$\cms$ estimates based FOLS and sa-OLS under different metrics. The bootstrap standard errors and confidence intervals within two standard deviations are given in parentheses.}
\label{fMRI_brain}
\resizebox{\textwidth}{!}{
\begin{tabular}{lcccccc}
\hline
Predictor   & FOLS (F)   & sa-OLS (F)  & FOLS (Chol F)  & sa-OLS (Chol F)  & FOLS (log F)  & sa-OLS (log F)     \\ \hline
Diagnosis [PD] & 0.6360 (0.3842)    & 0.6403 (0.0456)    & 0.8153 (0.3654)    & 0.7979 (0.0427)    & 0.8015 (0.3708)    & 0.7907 (0.0460)    \\
& (-0.1324, 1.4044)  & (0.5491,  0.7315)  & (0.0845,  1.5461)  & (0.7125,  0.8833)  & (0.0599,  1.5431)  & (0.6987,  0.8827)  \\ \hline
$\log$ Age  & -5.7125 (0.8634)   & -5.7132 (0.0419)   & -5.7654 (0.6876)   & -5.7660 (0.0404)    & -5.7636 (0.9290)   & -5.7626 (0.0406)   \\
& (-7.4393, -3.9857) & (-5.7970, -5.6294) & (-7.1406, -4.3902) & (-5.8468, -5.6852) & (-7.6216, -3.9056) & (-5.8438, -5.6814) \\ \hline
Sex [F] & 1.0446 (0.3873)    & 1.0437 (0.0016)    & 0.9286 (0.3633)    & 0.9378 (0.0032)    & 0.9396 (0.3708)    & 0.9469 (0.0050)    \\
& (0.2700, 1.8192)     & (1.0405 , 1.0469)  & (0.2020,  1.6552)  & (0.9314,  0.9442)  & (0.1980,  1.6812)  & (0.9369,  0.9569)  \\ \hline
\end{tabular}
}
\end{table}

In Table \ref{fMRI_brain}, we see that based on the FOLS estimate with the Frobenius norm, there is no significant difference in the correlation structure between the subjects diagnosed with Parkinson's disease and those without the diseases at a fixed age and sex. However, all other methods and metrics show significant difference. There is also variation in the level of the average difference depending on the metric choice. While both age and sex appear to be significantly related to the brain connectivity, the level of association also varies depending on the method and metric choice. Overall, the surrogate-assisted estimates appear to be more stable across metrics with smaller standard errors for the estimates compared to the Fr\'echet OLS estimates. The results in Table \ref{fMRI_brain} emphasizes the need to explore different estimators and metrics in applications.

\subsection{Glucose monitoring data}
Type 2 diabetes mellitus (T2DM) disease has become rampant across the world in recent decades. Classifying a patient as diabetic is typically based on some threshold of the level of glucose in the blood (glycaemia) or hemoglobin A1C, which could potentially lead to incorrect diagnosis. However, with recent advances in continuous glucose monitoring systems, we could improve the metrics for diagnosing type 2 diabetes mellitus. The goal of this study is to analyze the distribution of glycaemia and determine the factors influencing the distribution. This way, we take into account all aspects of the distribution including the center, the tails, and everything in-between.

Our data comes from the continuous glucose monitoring of 208 patients selected from the University Hospital of M\'ostoles, in Madrid, during an outpatient visit from January 2012 to May 2015. The data consists of the distribution of glycaemia recorded every 5 seconds for each patient over at least a 24 hour period. Patients were followed until either the diagnosis of T2DM or end of the study. Patients with basal glycaemia $\geq 126 mg/dL$ and/or haemoglobin A1c $\geq 6.5\%$, were diagnosed with T2DM at the end of the study. This data is publicly available and the details about the collection procedure can be found in \cite{colas2019detrended}. 

For our analysis, we take the glycaemia distributions as the response and clinical variables including age (years), body mass index (BMI; $kg/m^2$), basal glycaemia ($mg/dL$), and haemoglobin A1c (HbA1c; \%) as the predictors. Next, we proceed to estimate the $\cms$ using the Fr\'echet OLS (FOLS) and surrogate-assisted OLS (sa-OLS); and the $\spc$ using the Fr\'echet SIR (FSIR) and surrogate-assisted SIR (sa-SIR)  under different metrics, i.e., the 1-Wasserstein ($W_1$), 2-Wasserstein ($W_2$), and Hellinger ($H$) distances. 

To compare the performance of the methods and metrics, we estimate the leave-one-out average prediction error for the conditional Fr\'echet mean based on the estimated sufficient predictor using model (\ref{frechet_reg}). Let $\hat\beta$ denote the basis estimate. The average leave-one-out prediction error is given by
\begin{align}
    \delta = \cfrac{1}{n} \displaystyle\sum_{i=1}^n W_2^2\big(Y_i - \hat{m}_{\oplus}(\hat\beta^\top\bm x^{(-i)}) \big),
    \label{err_glucose}
\end{align}
where $\hat{m}_{\oplus}(\bm x)$ is the sample estimate of $m_{\oplus}(\bm x)$ and $\bm x^{(-i)}$ is $\bm x$ without the $i$th observation. The results for $\cms$ and $\spc$ estimates are provided in Tables \ref{glcy_cms} and \ref{glcy_spc}, respectively.
 
\begin{table}[htb!]
\centering
\caption{$\cms$ estimate based on different metrics}
\label{glcy_cms}
\resizebox{\textwidth}{!}{
\begin{tabular}{lcccccc}
\hline
& FOLS ($W_1$) & sa-OLS ($W_1$) & FOLS ($W_2$) & sa-OLS ($W_2$) & FOLS ($H$)  & sa-OLS ($H$) \\ \hline
Age  & -0.5746 & -0.5554 & -0.6745 & 0.6590 & -0.3691 & 0.4407 \\
BMI  & 0.0271 & 0.0418 & -0.1207 & 0.1055 & -0.1427 & 0.4233  \\
Basal glycaemia & -0.4789 & -0.4919 & -0.4202 & 0.4371 & 0.9241 & -0.9097 \\
HbA1c & -0.4283 & -0.4346 & -0.3658 & 0.3703 & -0.5121 & 0.2622 \\
\hline
$\delta$ & 11.9933 & 11.9840 & 11.9600 & 11.9756 & 12.0355 & 12.1097 \\ \hline
\end{tabular}
}
\end{table}

In Table \ref{glcy_cms}, we see very similar estimates based on the Wasserstein metrics as opposed to the Hellinger distance. Going by the estimates based on the Wasserstein metrics, the most important factor that influence the glycaemia distribution is Age, followed by basal glycaemia, then Hemoglobin A1C, and lastly BMI. However, based on the Hellinger distances, the most important determinant is the basal glycaemia, followed by age, then BMI, and lastly the hemoglobin A1C. In terms of prediction accuracy, the Wasserstein metrics appear to be better than the Hellinger estimates. We visualize the predicted distributions for the best and worse predictions in Figure \ref{fig:glyc_cms}.

 \begin{figure}[htb!]
 \captionsetup{width=\textwidth}
    \centering
    \includegraphics[width=0.49\linewidth, height=2.8in]{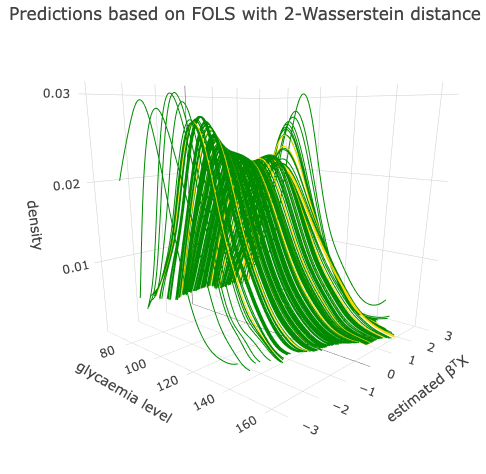}
    \includegraphics[width=0.48\linewidth, height=2.8in]{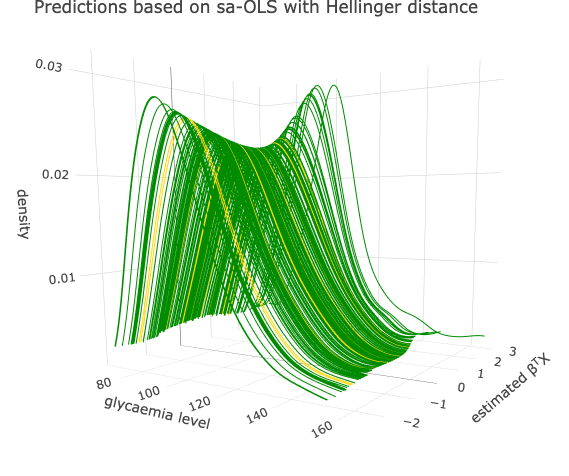}
    \caption{Predicted distributions of glycaemia based on $\hat\beta^\top\X$ for $\cms$ based on the 2-Wasserstein and Hellinger distances, respectively. The golden lines denote patients diagnosed with type 2 diabetes mellitus while the green lines denote otherwise.}
    \label{fig:glyc_cms}
\end{figure}

Similarly, in Table \ref{glcy_spc}, we provide the $\spc$ estimates. Again, we see conflicting results based on the Wasserstein metrics and the Hellinger distances. For the estimates based on the Wasserstein metrics, BMI is the most important, follow by age, then basal glycaemia, and then hemoglobin AIC. However, with the Hellinger distance, it is basal glycaemia, followed by hemoglobin A1C, then BMI, and lastly age. In terms of prediction, the 1-Wasserstein and Hellinger distances appear to perform better than the 2-Wasserstein. We visualize the surrogate-assisted SIR estimate based on the 1-Wasserstein and the Fr\'echet SIR estimate based on the Hellinger distance in Figure \ref{glcy_spc}.

\begin{table}[htb!]
\centering
\caption{$\spc$ estimate based on different metrics}
\label{glcy_spc}
\resizebox{\textwidth}{!}{
\begin{tabular}{lcccccc}
\hline
& FSIR ($W_1$) & sa-SIR ($W_1$) & FSIR ($W_2$) & sa-SIR ($W_2$) & FSIR ($H$)  & sa-SIR ($H$) \\ \hline
Age & -0.6268 & -0.6237 & -0.7085 & -0.6730 & -0.1396 & 0.0932 \\
BMI & -0.6546 & -0.6471 & -0.6767 & -0.6987 & 0.3607 & -0.4169 \\
Basal glycaemia & 0.5575 & 0.6133 & 0.4396 & 0.4989 & 0.6080 & -0.6510 \\
HbA1c & -0.3406 & -0.2906 & -0.2703 & -0.2442 & 0.5409 & -0.4496\\
 \hline
$\delta$ & 11.9795 & 11.9567 & 12.0099 & 12.0451 & 11.9848 & 11.9909 \\ \hline
\end{tabular}
}
\end{table}

Notice that variable importance varies depends on the estimator. The OLS estimators, i.e., FOLS and sa-OLS focus on the mean of the distribution, while the SIR estimators, i.e., FSIR and sa-SIR take the entire distributions into account. Thus, for the OLS estimators the tail variations in the distributions are not important, which is why the Hellinger distance performed worse at the predictions. On the other hand, SIR considers all aspects of the distributions, which is why the Hellinger distance becomes competitive and even performs better than the 2-Wasserstein. Another thing to notice is that estimators with significantly different basis estimates can yield similar prediction accuracy as seen in the 1-Wasserstein and Hellinger SIR estimates. Finally, by examining the estimated distributions, we see that the distributions of patients diagnosed with T2DM do not necessarily standout compared to those without the disease.

\begin{figure}[htb!]
\captionsetup{width=\textwidth}
    \centering
    \includegraphics[width=0.49\linewidth, height=2.8in]{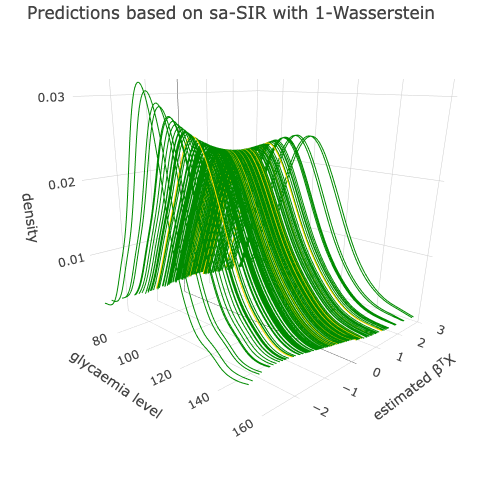}
    \includegraphics[width=0.48\linewidth, height=2.8in]{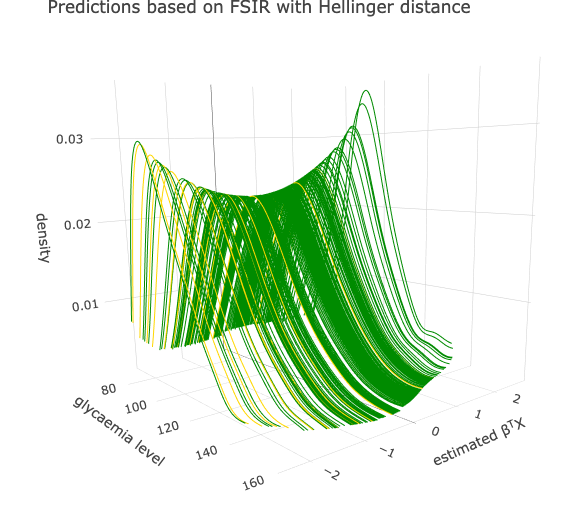}

\caption{Predicted distributions of glycaemia based on $\hat\beta^\top\X$ for $\spc$ based on the 1-Wasserstein and Hellinger distances, respectively. The golden lines denote patients diagnosed with type 2 diabetes mellitus while the green lines denote otherwise.}
\label{fig:glyc_spc}
\end{figure}

\section{Discussion}
In this paper, we investigated how metric choice plays a crucial role in Fr\'echet regression. Across metric spaces, we saw that while some metrics are very popular in the literature, there is no single metric that can serve all purposes. Whether or not a metric is appropriate for a given metric space depends on the specific application. In particular, for estimating the conditional Fr\'echet mean, which is the focus of most applications, metric choice is key.  

The conflicting conclusions from the $\cms$ and $\spc$ estimates based on different metrics buttresses the point that metric choice plays a crucial role, especially when dealing with the mean space, as we observed in the numerical studies. A future study that explores how metric choice affects the sufficient dimension reductions estimates in regression settings where both the predictor and response to belong to random metric spaces will be a reasonable extension.








\bibliographystyle{agsm}
\bibliography{ref}

\appendix

\section*{Appendix of Proofs}

\begin{proof}[Proof of Proposition \ref{lp_ordering}]\quad \\
 Let $(\bOmega_Y, \rho, \cP)$ denote a measurable metric space with $\cP$ as the probability measure on $\bOmega_Y$. Also, let $1 \leq p < 2 < q < \infty$. Suppose a continuous function $f \in \mathcal{L}^2$, then $\displaystyle\int_{\bOmega_Y} |f|^2 d\cP < \infty$. Now, since $0 < p/2 < 1$, by H\"older's inequality, we have
 \begin{align*}
 \lVert f\rVert_2 = \displaystyle\int_{\bOmega_Y} |f|^2 d\cP &= \displaystyle\int_{\bOmega_Y} |f|^2 .1 d\cP \\
 &\leq  \left(\displaystyle\int_{\bOmega_Y} |f|^{2p/2}d\cP\right)^{p/2} \left(\displaystyle\int_{\bOmega_Y} d\cP \right)^{1-p/2} \\
 &= \left(\displaystyle\int_{\bOmega_Y} |f|^p d\cP\right)^{p/2}\cP(\bOmega_Y)^{1-p/2}
 \end{align*}
By setting $\cP(\bOmega_Y)=1$, we have $\lVert f\rVert_2 \leq \lVert f \rVert_p$.

Notice that we also have $0 < \frac{1}{q} < \frac{1}{2} < \frac{1}{p} \leq 1$. Thus, there exist a unique $\alpha \in (0, 1)$ such that $\cfrac{1}{2} = \cfrac{\alpha}{p} + \cfrac{1-\alpha}{q}$. Therefore, by Littlewood's inequality, we have $\lVert f \rVert_2 \leq  \lVert f \rVert^{\alpha}_p \lVert f \rVert^{1-\alpha}_q$, which implies $\lVert f \rVert_2 \leq \lVert f \rVert_q$.\\
\end{proof}

\begin{proof}[of Proposition \ref{Wasserstein_ordering}]\quad \\
The proof follows directly by Lyapunov inequality: $(\E\lvert y - y'\rvert^{p_1} \big )^{1/p_1}  \leq \big (\E\lvert y - y'\rvert^{p_2} \big )^{1/p_2}$ for $0 < p_1 < p_2 < \infty$. The desired results is achieved by setting $p_1=1$ and $p_2 = p > 1$ in (\ref{wasserstein}). \\
\end{proof}

\begin{proof}[Proof of Proposition \ref{stochastic_dominance}] \quad \\
$Y_1 \succeq_1 Y_2$ if and only if $F_{Y_1}(t) \leq F_{Y_2}(t), \forall t\in [0,\infty)$. Thus, for any non-decreasing function $g(y)$, we have $\E_{Y_1}(g(y)) \geq \E_{Y_2}(g(y))$. Therefore, $W_1(Y_1, Y_2) = \E|Y_1 - Y_2| = \E[\E|Y_1 - Y_2|\mid Y_2 < Y_1, Y_2] = \E(\E(Y_1 - Y_2\mid Y_2 < Y_1, Y_2)) = \E(Y_1) - \E(Y_2)$.    
\end{proof}

\begin{proof}[Proof of Proposition \ref{spdm_prop}] \quad \\
For any $y, y' \in  \bm S_r^{+}$,
\begin{align*}
\lVert  y - y' \rVert_F  &=  \lVert  \L_y\L_y^\top -  \L_{y'}\L_{y'}^\top \rVert_F \\
&=  \lVert  \L_y\L_y^\top - \L_{y}\L_{y'}^\top + \L_{y}\L_{y'}^\top -  \L_{y'}\L_{y'}^\top \rVert_F  \\
&= \lVert  \L_y(\L_y^\top - \L_{y'}^\top) + (\L_{y} -  \L_{y'})\L_{y'}^\top \rVert_F \\
& \leq \lVert  \L_y (\L_y^\top - \L_{y'}^\top) \rVert_F + \lVert (\L_{y} -  \L_{y'})\L_{y'}^\top \rVert_F, \ \text{ by the triangle inequality of } \lVert.\rVert_F \\
& \leq \lVert \L_y \rVert_F \lVert \L_y^\top - \L_{y'}^\top \rVert_F + \lVert \L_{y} -  \L_{y'} \rVert_F \lVert \L_{y'}^\top \rVert_F,  \ \text{ by the sub-multiplicative property of } \lVert.\rVert_F\\
&\leq 2\lambda_{\max}\lVert \L_{y} -  \L_{y'} \rVert_F,
\end{align*}
 where $\lambda_{\max} = \max\{\lVert \L_y \rVert_F, \lVert \L_{y'} \rVert_F \} $. The desired results is achieved for $\lambda_{\max} \geq 0.5$. \\  
\end{proof}

\begin{proof}[Proof of Theorem \ref{containment}] \quad \\
By Theorem 2.3 of \cite{li2018sufficient}, the central space is such that $\mathcal{S}_{\psi(Y)\mid\X} \subseteq \spc$, for any measurable function $\psi(.)$. Thus, 
\begin{align}
    \left\{\mathcal{S}_{f(Y)\mid\X} : f\in \mathcal{L}_{\rho_2}(\bOmega_Y)  \right\} \supseteq  \left\{\mathcal{S}_{f(Y)\mid\X} : f\in \mathcal{L}_{\rho_1}(\bOmega_Y) \right\}.
\end{align}
Therefore, 
\begin{align}
    \bigcap\left\{\mathcal{S}_{f(Y)\mid\X} : f\in \mathcal{L}_{\rho_2}(\bOmega_Y)  \right\} \subseteq  \bigcap\left\{\mathcal{S}_{f(Y)\mid\X} : f\in \mathcal{L}_{\rho_1}(\bOmega_Y)  \right\}.
\end{align}
\end{proof}

\begin{proof}[Proof of Theorem \ref{lipschitz_cms}]\quad \\
Suppose the embedding based on $\rho_1$ concentrates around the mean. Without loss of generality, we let the surrogate response $\widetilde Y_{\rho_1} = h(\bEta^\top\X) + \epsilon$, where $h(.)$ is some unknown link function and $\E(\epsilon) = 0$. Then, $\E(\widetilde Y_{\rho_1}|\X) \indep \X | \bEta^\top\X$. Therefore, the $\mathcal{S}_{\E(\widetilde Y_{\rho_1}|\X)}$ estimate captures all the regression information.

Now, suppose $\widetilde Y_{\rho_2} = h_1(\bEta_1^\top\X) + h_2(\bEta_2^\top\X)\epsilon$, where $\bEta = (\bEta_1, \bEta_2)\in\R^d$ is such that $\bEta^\top\bEta = \I_d$. Then $\mathcal{S}_{\E(\widetilde Y_{\rho_2}|\X)}  = \mathrm{span}(\bEta_1)$ while $\mathcal{S}_{\widetilde Y_{\rho_2}|\X} = \mathrm{span}(\bEta_1,\bEta_2)$. Therefore, estimating $\mathcal{S}_{\E(\widetilde Y_{\rho_2}|\X)}$ leads to some loss of regression information.
\end{proof}

\section*{Hellinger distances}
{\bf Homogeneous Poisson distributions}\\
Let $P(\lambda_1)$ and $P(\lambda_2)$ denote two homogeneous Poisson distributions with parameters $\lambda_1, \lambda_2 > 0$. Then the Hellinger distance between them is given by 
\begin{align*}
    H^2[P_\omega(\lambda_1), P_\omega(\lambda_2)] &= 1 - \displaystyle\sum_{\omega=0}^\infty \sqrt{P_\omega(\lambda_1)P_\omega(\lambda_2)} \\ 
    &=  1 - \displaystyle\sum_{\omega=0}^\infty \cfrac{e^{-(\lambda_1 + \lambda_2)/2}(\sqrt{\lambda_1\lambda_2})^\omega}{\omega!}  \\
    &= 1 - e^{-(\lambda_1 + \lambda_2)/2}\displaystyle\sum_{\omega=0}^\infty \frac{(\sqrt{\lambda_1\lambda_2})^\omega}{\omega!} \\
     &= 1 - e^{-(\lambda_1 + \lambda_2)/2 + \sqrt{\lambda_1\lambda_2}}
     = 1 - e^{-\frac{1}{2}(\sqrt{\lambda_1} - \sqrt{\lambda_2})^2}
\end{align*}


{\bf Gamma distributions with the same scale parameter} \\
Consider two Gamma distributions with shape and scale parameters $(\alpha_1, \theta)$ and $(\alpha_2,\theta)$, respectively, where $\alpha_1, \alpha_2, \theta > 0$. Then the Hellinger distance between them is given by
\begin{align*}
    H^2[f_\omega(\alpha_1, \theta), f_\omega(\alpha_2, \theta)] &= 1 - \int \sqrt{f_\omega(\alpha_1, \theta)f_\omega(\alpha_2, \theta)}  \partial \omega \\
    &= 1 - \int_0^\infty \cfrac{\omega^{(\alpha_1+\alpha_2)/2-1}e^{-\omega/\theta}}{\sqrt{\Gamma(\alpha_1)\Gamma(\alpha_2)\theta^{\alpha_1+
    \alpha_2}}}  \partial \omega \\
    &= 1-\cfrac{1}{\sqrt{\Gamma(\alpha_1)\Gamma(\alpha_2)\theta^{\alpha_1+
    \alpha_2}}} \int_0^\infty \omega^{(\alpha_1+\alpha_2)/2-1}e^{-\omega/\theta}  \partial \omega \\
    &= 1-\cfrac{\theta^{(\alpha_1+\alpha_2)/2}\Gamma\left(\frac{\alpha_1+\alpha_2}{2}\right)}{\sqrt{\Gamma(\alpha_1)\Gamma(\alpha_2)\theta^{\alpha_1+
    \alpha_2}}} 
    = 1-\cfrac{\Gamma\left(\frac{\alpha_1+\alpha_2}{2}\right)}{\sqrt{\Gamma(\alpha_1)\Gamma(\alpha_2)}}\\
\end{align*}

{\bf Log-normal distributions with the same scale parameter} \\
Consider two log-normal distributions with location and scale parameters $(\mu_1, \sigma)$ and $(\mu_2,\sigma)$, respectively, where $(\mu_1, \mu_2) \in (-\infty, \infty)$ and $\sigma > 0$. Then the Hellinger distance between them is given by
\begin{align*}
    H^2[f_\omega(\mu_1, \sigma), f_\omega(\mu_2, \sigma)] &= 1 - \int \sqrt{f_\omega(\mu_1, \sigma)f_\omega(\mu_2, \sigma)}  \partial \omega \\
    &= 1 - \int_0^\infty \cfrac{1}{\omega\sigma\sqrt{2\pi}}e^{-\frac{1}{2\sigma^2}[(\ln \omega-\mu_1)^2/2 + (\ln \omega-\mu_2)^2/2] } \partial \omega \\
    &= 1 - e^{-\frac{(\mu_1-\mu_2)^2 }{8\sigma^2}} \int_0^\infty \cfrac{1}{\omega\sigma\sqrt{2\pi}}e^{-\frac{1}{2\sigma^2}(\ln \omega-\frac{\mu_1+\mu_2}{2})^2 } \partial \omega \\
    &= 1 - e^{-\frac{(\mu_1-\mu_2)^2 }{8\sigma^2}}
\end{align*}

\end{document}